\documentclass[twocolumn,3p,draft]{elsarticle}
\usepackage{amsmath}
\usepackage{amssymb}

\parskip\smallskipamount

\let\set\mathbb
\def\<#1>{\langle#1\rangle}

\newtheorem{prop}{Proposition}
\newtheorem{defn}{Definition}

\newenvironment{proof}{\par\noindent\textbf{Proof.}}{\par\smallskip}

 \def\applyR#1#2#3{$#1$,\kern.7em$#2$\kern.7em$\stackrel{\mathrm{R}}\longrightarrow$\kern.7em$#3$}
 \def\applyS#1#2#3{$#1$,\kern.7em$#2$\kern.7em$\stackrel{\mathrm{S}}\longrightarrow$\kern.7em$#3$}
 \def\applyU#1#2{$#1$\kern.7em$\stackrel{\mathrm{U}}\longrightarrow$\kern.7em$#2$}

\begin{document}

\title{Short Proofs for Some Symmetric Quantified Boolean Formulas}

\author[mk]{Manuel Kauers}
\ead{manuel.kauers@jku.at}

\author[ms]{Martina Seidl}
\ead{martina.seidl@jku.at}

\address[mk]{Institute for Algebra, J. Kepler University Linz, Austria}
\address[ms]{Institute for Formal Models and Verification, J. Kepler University Linz, Austria}

 \begin{abstract}
  We exploit symmetries to give short proofs for two prominent formula families 
  of QBF proof complexity. On the one hand, we employ symmetry 
  breakers. On the other hand, we enrich the (relatively weak) 
  QBF resolution calculus 
  Q-Res with the symmetry rule and obtain separations to powerful QBF calculi. 
 \end{abstract}


 \maketitle

 \section{Introduction}\label{sec:0}

 A Quantified Boolean Formula (QBF) is a formula of the form $P.\phi$, where
 $\phi$ is a propositional formula, 
 say in the variables $x_1,\dots,x_n$, and $P$
 is a quantifier prefix $P=Q_1x_1Q_2x_2\cdots Q_nx_n$ with $Q_i\in\{\forall,\exists\}$.
 From QBF proof complexity, it is well-known that the shortest 
 proof of certain QBFs 
 may have
 exponential size in a resolution-based calculus~\cite{DBLP:books/daglib/0075409,DBLP:conf/stacs/BeyersdorffCJ15}. 
 We consider here two families of QBFs (cf.~Section~\ref{sec:1})
 which play a prominent role in QBF proof complexity for separating 
 various calculi. We make the 
 observation that short proofs can be obtained if we
 take into account the symmetries of the formulas. 
 In Section~\ref{sec:2}, we do so by
 using symmetry breakers. In Section~\ref{sec:3}, we enrich
 the oldest variant of the resolution calculus for QBF, Q-Res~\cite{DBLP:journals/iandc/BuningKF95},
 by a \emph{symmetry rule}, generalizing
 an idea reported in~\cite{DBLP:journals/acta/Krishnamurthy85,DBLP:journals/dam/Urquhart99} for SAT. In both cases, it turns out that the proof sizes
 for both families of formulas shrinks from exponential to linear. 
 As consequences, we obtain 
 separation results between Q-Res with the symmetry rule and 
 powerful proof systems like IR-calc~\cite{DBLP:conf/stacs/BeyersdorffCJ15}
and LQU$^+$~\cite{DBLP:conf/sat/BalabanovWJ14} (cf.~Section~\ref{sec:6}).

 Let us recall some basic facts and fix some notation. 
 We only consider QBFs $P.\phi$ where $\phi$ is in conjunctive normal form (CNF), i.e., $\phi$ is a
 conjunction of clauses, each clause being a disjunction of literals, each literal being a variable
 or a negated variable, i.e., if $x$ is a variable, $x$ and $\bar x$ are 
 literals. We also view clauses as sets of literals. 
 The prefix $P = Q_1x_1\ldots Q_nx_n$ imposes an order $<_P$ on its variables: 
 $x_i <_P x_j$ if $i < j$. The Q-Res calculus~\cite{DBLP:journals/iandc/BuningKF95}
 applies the following rules on a QBF $P.\phi$: 
 \begin{enumerate}
 \item[A] Any clause of $\phi$ can be derived.
 \item[R] From the already derived clauses $C\lor x$ and $C'\lor\bar x$ 
   with existentially quantified variable $x$ and
   $C,C'$ such that $C\cup C'$ is not a tautology,
   the clause $C\lor C'$ can be derived.
 \item[U] Let $C\lor l$ be an already derived clause where $l$ is 
 a universal literal, $\bar l \not\in C$ 
 and all existential literals $k \in C$ 
 are such that $k <_P l$. Then the clause $C$ can be derived. 
 \end{enumerate}

 In the following, we do not mention the application of the axiom rule~A
 explicitly. We write  \applyR{C_1}{C_2}{C} and \applyU{D_1}{D}
 for the application of R and~U. 
 A refutation of a QBF $P.\phi$ is the consecutive application of the 
resolution rule R and the universal reduction rule U until the empty 
clause is derived. Q-Res is sound and complete. 

Finally, let us recall the notion of (syntactic) symmetries for QBFs. A bijective map $\sigma$ from
 the set $\{x_1,\dots,x_n,\bar x_1,\dots,\bar x_n\}$ of literals to itself
 is called admissible for a prefix $P=Q_1x_1\dots Q_nx_n$
 if $\overline{\sigma(x)}\leftrightarrow\sigma(\bar x)$ for all $x\in\{x_1,\dots,x_n\}$
 and for all $i,j\in\{1,\dots,n\}$, we have $\sigma(x_i)\in\{x_j,\bar x_j\}$ only if
 $x_i$ and $x_j$ belong to the same quantifier block, i.e., $Q_{\min(i,j)}=\cdots=Q_{\max(i,j)}$.
 An admissible function $\sigma$ is called a symmetry for a QBF $P.\phi$ with $\phi$ in CNF
 if applying $\sigma$ to all literals in $\phi$ maps $\phi$ to itself (possibly up to
 reordering clauses and literals).

 \section{Formula Families}\label{sec:1}

 We consider the following two families of formulas.

 \def\kbkf{\mathrm{KBKF}}
 \begin{defn}[\cite{DBLP:books/daglib/0075409}]
  \label{def:hkb}
   For $n\in\set N$, the formula $\kbkf_n$ is defined by the prefix
   \[
   \exists x_1y_1\forall a_1
   \exists x_2y_2\forall a_2\dots
   \exists x_ny_n\forall a_n
   \exists z_1\dots z_n
   \]
   and the following clauses:
   \begin{itemize}
   \item $C_1=(\bar x_1\lor\bar y_1)$
   \item 
     for $j=1,\dots,n-1$:

     $C_{2j} =(x_j\lor\bar a_j\lor\bar x_{j+1}\lor\bar y_{j+1})$\\
     $C_{2j+1}=(y_j\lor a_j\lor\bar x_{j+1}\lor\bar y_{j+1})$.
   \item $C_{2n} = (x_n\lor \bar a_n\lor\bar z_1\lor \ldots \bar z_n)$,\\
    $C_{2n+1} = (y_n\lor a_n\lor\bar z_1\lor \ldots \bar z_n)$
   \item for $j=1,\dots,n$:

     $B_{2j-1}=(a_j\lor z_j)$ and $B_{2j}=(\bar a_j\lor z_j)$.
   \end{itemize}   
 \end{defn}

 For every $n\in\set N$, the formula $\kbkf_n$ is false, and it is 
 known~\cite{DBLP:books/daglib/0075409} that 
 any Q-Res refutation needs a number of steps which is at least exponential in~$n$.

 \def\parity{\mathrm{QUPARITY}}
 \def\qparity{\mathrm{QPARITY}}

 \begin{defn}[\cite{DBLP:conf/stacs/BeyersdorffCJ15}]
   For $n\in\set N$ with $n > 1$, the formula $\parity_n$ is defined by the prefix
   \[
   \exists x_1\dots x_n\forall a_1a_2\exists y_2\dots y_n
   \]
   and the following clauses:
   \begin{itemize}
   \item $A_2=(\bar x_1\lor\bar x_2\lor\bar y_2 \lor a_1 \lor a_2)$\\
     $B_2=(\bar x_1\lor x_2\lor y_2 \lor a_1 \lor a_2)$\\
     $C_2=(x_1\lor\bar x_2\lor y_2 \lor a_1 \lor a_2)$\\
     $D_2=(x_1\lor x_2\lor\bar y_2 \lor a_1 \lor a_2)$

   \item  for $j=3,\dots,n$:

     $A_j=(\bar y_{j-1}\lor\bar x_j\lor\bar y_j \lor a_1 \lor a_2)$\\
     $B_j=(\bar y_{j-1}\lor x_j\lor y_j \lor a_1 \lor a_2)$\\
     $C_j=(y_{j-1}\lor\bar x_j\lor y_j \lor a_1 \lor a_2)$\\
     $D_j=(y_{j-1}\lor x_j\lor\bar y_j \lor a_1 \lor a_2)$
   \item $E_1=(a_1 \lor a_2 \lor y_n)$ and $E_2=(\bar a_1 \lor \bar a_2 \lor\bar y_n)$

   \item for $i=2,\dots,n$, $A'_i, B'_i, C'_i, D'_i$ are obtained 
   from $A_i, B_i, C_i, D_i$ by replacing $a_1 \lor a_2$ by $\bar a_1 \lor \bar a_2$. 
   \end{itemize}
 \end{defn}

 $\parity_n$ is a variant of the $\qparity_n$ family~\cite{DBLP:conf/stacs/BeyersdorffCJ15}
 which encodes 
 $\exists x_1\dots x_n\forall z. z\not=x_1\oplus\cdots\oplus x_n$, 
 where $\oplus$ stands for exclusive or.
 Obviously all these formulas are false. Refuting $\qparity_n$
 needs an exponential number of steps in the calculus Q-Res, but
 not in the stronger calculus LQU$^+$. We use $\parity_n$ instead of $\qparity_n$
 because for this family, also LQU$^+$ needs exponentially many steps~\cite{DBLP:conf/stacs/BeyersdorffCJ15}.
 This will be used in Section~\ref{sec:6}.

 \section{Symmetry Breakers}\label{sec:2}

 Let $S$ be a set of symmetries for a QBF~$P.\phi$.
 A symmetry breaker is a certain Boolean formula $\psi$ such that when $P.\phi$ is true, so is $P.(\phi\land\psi)$.
 Writing $P=Q_1x_1\cdots Q_nx_n$, it was shown in~\cite{audemard2007efficient,DBLP:journals/corr/abs-1802-03993} that 
 \[
 \psi=\bigwedge_{\vbox{\hbox to0pt{\hss$\scriptstyle i=1$\hss}\kern-3pt\hbox to0pt{\hss$\scriptstyle Q_i=\exists$\hss}}}^n
      \ \ \bigwedge_{\sigma\in S}
      \biggl(\Bigl(\bigwedge_{j<i} (x_j\leftrightarrow\sigma(x_j))\Bigr)\rightarrow(x_i\rightarrow\sigma(x_i))\biggr)
 \]
 is a symmetry breaker.

 For the formulas $\kbkf_n$ (Def.~\ref{def:hkb}), 
 we have for every $i=1,\dots,n$
 the symmetry $\sigma_i=(x_i\ y_i)(\bar x_i\ \bar y_i)(a_i\ \bar a_i)$ which exchanges
 the variables $x_i, y_i$, the literals $\bar x_i,\bar y_i$, and the literals $a_i,\bar a_i$.
 Therefore,
 \[
   \psi_n = (\bar x_1\lor y_1)\land\cdots\land(\bar x_n\lor y_n)
 \]
 is a symmetry breaker for~$\kbkf_n$.

 \begin{prop}
   For $n\in\set N$, write $\kbkf_n$ as $P_n.\phi_n$, and let $\psi_n$ be the symmetry breaker from above.
   Then $P_n.(\phi_n\land\psi_n)$ has a refutation proof with no more than 
   $4n$ steps. 
 \end{prop}

 The proof proceeds as follows.
 \begin{itemize}
 \item \applyR{C_1}{(\bar x_1\lor y_1)}{U_0:=\bar x_1}.

 \kern-\smallskipamount
 \item for $j=1,\dots,n-1$, do

   \applyR{C_{2j}}{U_{j-1}}{\tilde U_{j}:=(\bigvee_{i=1}^j\bar a_i{\lor}\bar x_{j+1}{\lor}\bar y_{j+1})}.

   \applyR{\tilde U_{j}}{(\bar x_{j+1}{\lor}y_{j+1})}{U_{j}:=(\bigvee_{i=1}^j\bar a_i{\lor}\bar x_{j+1})}.

   Then $U_{n-1}=(\bar a_1\lor\dots\lor\bar a_{n-1}\lor\bar x_n)$.
   
 \kern-\smallskipamount
 \item \applyR{C_{2n}}{U_{n-1}}{V_0:=(\bigvee_{i=1}^n\bar a_i\lor\bar z_1\lor\dots\lor\bar z_n)}.

 \kern-\smallskipamount
 \item for $j=1,\dots,n$, do

   \applyR{V_{j-1}}{B_{2j}}{V_j:=(\bigvee_{i=1}^n\bar a_i\lor\bigvee_{i=j+1}^n\bar z_i)}.
   
   Then $W_0:=V_n=(\bar a_1\lor\dots\lor\bar a_n)$.

 \kern-\smallskipamount
 \item for $j=1,\dots,n$, do

   \applyU{W_{j-1}}{W_j:=(\bar a_{j+1}\lor\dots\lor\bar a_n)}.
   
   $W_n$ is the empty clause. 
 \end{itemize}

 For the formulas $\parity_n$, the argument is similar. In this case,
 we have the symmetries $\sigma_1 = (x_1\ x_2)(\bar x_1\ \bar x_2)$ and
 \[
 \sigma_i=(x_i\ \bar x_i)(a_1\ \bar a_1)(a_2\ \bar a_2)(y_i\ \bar y_i)\cdots(y_n\ \bar y_n)
 \]
 for every $i=2,\dots,n$.
 There are some further symmetries which we will not need.
 The symmetries $\sigma_1,\dots,\sigma_n$ give rise to the
 symmetry breaker
 \[
   \psi_n=(\bar x_1 \lor x_2) \land \bar x_2\land\dots\land \bar x_n
 \]
 for $\parity_n$.

 \begin{prop}
   For $n\in\set N$ with $n > 1$, write $\parity_n$ as $P_n.\phi_n$, and let $\psi_n$ be the symmetry breaker from above.
   Then $P_n.(\phi_n\land\psi_n)$ has a refutation proof with no more than $2n+1$ steps. 
 \end{prop}

 The proof proceeds as follows.
 \begin{itemize}
 \item \applyR{D_2}{(\bar x_1\lor x_2)}{U_1:=(x_2 \lor \bar y_2 \lor a_1 \lor a_2)}.

 \kern-\smallskipamount
 \item \applyR{U_1}{\bar x_2}{U_2 := (\bar y_2 \lor a_1 \lor a_2)}.

 \kern-\smallskipamount
 \item for $j=3,\dots,n$, do

   \applyR{D_j}{\bar x_j}{\tilde D_j:=(y_{j-1}\lor\bar y_j \lor a_1 \lor a_2)}.

 \kern-\smallskipamount
 \item for $j = 3,\dots,n$, do

   \applyR{U_{j-1}}{\tilde D_j}{U_j:=(\bar y_j \lor a_1 \lor a_2)}.

 \kern-\smallskipamount
 \item \applyR{U_n=(\bar y_n \lor a_1 \lor a_2)}{E_1}{(a_1\lor a_2)}.

 \kern-\smallskipamount
 \item \applyU{(a_1\lor a_2)}{\text{\applyU{a_2}{\text{empty clause}}}}.
 \end{itemize}
  
 \section{The Symmetry Rule}\label{sec:3}

 As an alternative to using symmetry breakers, we can enrich the calculus Q-Res 
 as introduced in Section~\ref{sec:0} to the calculus Q-Res+S by adding the
 following rule, which allows us to exploit symmetries of the input formula $P.\phi$
 within the proof. 
 \begin{enumerate}
 \item[S] From an already derived clause $C$ and a symmetry $\sigma$ of $P.\phi$,
 the clause $\sigma(C)$ can be derived.
 \end{enumerate}
 Several variants of this rule have been proposed 
 for SAT in~\cite{DBLP:journals/acta/Krishnamurthy85,DBLP:journals/dam/Urquhart99}, but to our knowledge it has not yet been
 considered in the context of QBF.
 However, it is easy to see that the rule also works for QBF.
 
 \begin{prop}
   Let $P.\phi$ be a QBF, and suppose that $C$ is a clause which can be derived from $\phi$
   using the rules S, R,~U.
   Then it can also be derived using only the rules R,~U. 
 \end{prop}
 \begin{proof}
   Suppose otherwise. Then there are clauses which can be derived with S, R,~U but not with R,~U alone.
   Let $C$ be such a clause, and consider a derivation of $C$ with a minimal number of applications
   of~S. The rule S is used at least once during the derivation. Consider its earliest application,
   suppose this application derives $\sigma(D)$ from the clause~$D$.
   If we can show that $\sigma(D)$ can also be derived using only R and~U, then we can eliminate
   this first application of S in the derivation of~$C$ and obtain a contradiction to the assumed
   minimality.

   To show that $\sigma(D)$ can be derived using only R and~U, observe first that $D$ was derived
   only using R and~U.
   For an admissible function $\sigma$, we have $\overline{\sigma(x)}\leftrightarrow\sigma(\bar x)$ for every variable~$x$.
   Therefore, if a clause $E$ can be derived by R from two clauses $E_1$ and~$E_2$, we can derive $\sigma(E)$
   by R from $\sigma(E_1)$ and $\sigma(E_2)$.
   Furthermore, an admissible function cannot permute literals across quantifier blocks, which implies
   that if $F$ can be derived by U from $F_1$, then $\sigma(F)$ can be derived by U from $\sigma(F_1)$.
   Finally, when $\sigma$ is a symmetry of $\phi$ and $G$ is a clause of~$\phi$, then also $\sigma(G)$
   is a clause of~$\phi$.
   By combining these three observations, it follows that applying $\sigma$ to all clauses appearing
   in the derivation of $D$ yields a derivation of~$\sigma(D)$.
   This completes the proof.
   \qed
 \end{proof}

 According to the previous proposition, with S we cannot derive 
 any clause that we cannot also
 derive without.
 Therefore, soundness of Q-Res+S follows from soundness of Q-Res.
 Next, we illustrate that Q-Res+S allows for shorter proofs than Q-Res.
 For the application of S, we write \applyS{C}{\sigma}{D}.

 \begin{prop}
   For every $n\in\set N$, the formula $\kbkf_n$ can be refuted by no more than $5n$ applications of S, R,~U.
 \end{prop}

 We proceed as follows by using the symmetries of the form $\sigma_i=(x_i\ y_i)(\bar x_i\ \bar y_i)(a_i\ \bar a_i)$ for $i=1,\dots,n$.
 \begin{itemize}
 \item set $U_{n+1}=C_{2n+1}$.

 \kern-\smallskipamount
 \item for $j=n,\dots,1$, do

   \applyR{U_{j+1}}{B_{2j-1}}{U_j:=(y_n\lor\bigvee_{i=j}^na_i\lor\bigvee_{i=1}^{j-1}\bar z_i)}.

 \kern-\smallskipamount
 \item set $W_n:=U_1=(y_n\lor a_1\lor\dots\lor a_n)$.

 \kern-\smallskipamount
 \item for $j=n,\dots,2$, do

   \applyU{W_j}{V_j:=(y_j\lor\bigvee_{i=1}^{j-1} a_i)}.\\
   \applyS{V_j}{\sigma_j}{V_j':=(x_j\lor\bigvee_{i=1}^{j-1} a_i)}.\\
   \applyR{V_j'}{C_{2j-1}}{V_j'':=(y_{j-1}\lor\bar x_j\lor\bigvee_{i=1}^{j-1} a_i)}.\\
   \applyR{V_j''}{V_j}{W_{j-1}:=(y_{j-1}\lor\bigvee_{i=1}^{j-1} a_i)}.

 \kern-\smallskipamount
 \item \applyU{W_1=(y_1\lor a_1)}{V_1=y_1}.

 \kern-\smallskipamount
 \item \applyS{V_1}{\sigma_1}{V_1':=x_1}.

 \kern-\smallskipamount
 \item \applyR{V_1'}{C_1}{V_1'':=\bar y_1}.

 \kern-\smallskipamount
 \item \applyR{V_1''}{V_1}{\text{empty clause}}.
 \end{itemize} 
 
 \begin{prop}
   For every $n\in\set N$ with $n > 1$, the formula $\parity_n$ can be refuted by no more than $3n+2$ applications of S, R,~U.
 \end{prop}

 Recall from Section~\ref{sec:3} that $\parity_n$ has the symmetries $\sigma_1=(x_1\ x_2)(\bar x_1\ \bar x_2)$ and  $\sigma_i=(x_i\ \bar x_i)(a_1\ \bar a_1)(a_2\ \bar a_2)(y_i\ \bar y_i)\cdots(y_n\ \bar y_n)$ 
 for $i > 1$.
 \begin{itemize}
 \item \applyR{D_n}{E_1}{U_n:=(y_{n-1}\lor x_n\lor a_1 \lor a_2)}.

 \kern-\smallskipamount
 \item for $j=n-1,\dots,3$, do

   \applyR{D_j}{U_{j+1}}{U_j:=(y_{j-1}{\lor}\bigvee_{i=j}^n x_i{\lor}a_1{\lor}a_2)}.\kern-5pt\null

 \kern-\smallskipamount
 \item \applyR{D_2}{U_3}{U_2:=(\bigvee_{i=1}^n x_i\lor a_1 \lor a_2)}.

 \kern-\smallskipamount
 \item \applyU{U_2}{\text{\applyU{\bigvee_{i=1}^n x_i\lor a_1}{V_n:=\bigvee_{i=1}^n x_i}}}.

 \kern-\smallskipamount
 \item for $j=n,\dots,2$, do

   \applyS{V_j}{\sigma_j}{W_j:=(x_1\lor\dots\lor x_{j-1}\lor\bar x_j)}.\\
   \applyR{V_j}{W_j}{V_{j-1}:=(x_1\lor\dots\lor x_{j-1})}.

 \kern-\smallskipamount
 \item \applyS{V_1=x_1}{\sigma_1}{W_1:=x_2}.

 \kern-\smallskipamount
 \item \applyS{W_1}{\sigma_2}{W_2:=\bar x_2}.

 \kern-\smallskipamount
 \item \applyR{W_1}{W_2}{\text{empty clause}}.
 \end{itemize} 
 
 \section{Consequences}\label{sec:6}
 
 From recent results, it is known that plain Q-Res is rather weak
(for a fine-grained comparison of QBF 
 proof systems see~\cite{DBLP:conf/stacs/BeyersdorffCJ15}). Both, 
 the expansion-based proof system IR-calc and the CDCL-based proof system
 LQU$^+$ are strictly stronger than Q-Res. 
 The addition of the symmetry rule changes the situation. While 
 the $\parity_n$ formulas are hard for LQU$^+$ and the $\kbkf_n$ formulas 
 are hard for IR-calc, we have shown that both are easy for Q-Res+S. 
 Now one may ask
 if Q-Res+S is strictly stronger than IR-calc or LQU$^+$. The
 answer is clearly ``no''. For $\kbkf_n$, the application of the 
 symmetry rule can be hindered by introducing $n$ universally quantified 
 variables $b_i$ which are placed between $x_i$ and $y_i$ in the prefix.
 Further, each clause $C_{2j}$  changes to $C_{2j} \vee b_j$.   
 For this modified formula, LQU$^+$ can still find a short proof, 
 but Q-Res+S can only apply R and U, hence it falls back to Q-Res 
 which does not exhibit short proofs for $\kbkf_n$. 
 In a similar way, $\parity_n$ can be modified such that these 
 formulas remain simple for IR-calc, but become hard for Q-Res+S.

 \begin{prop}
 Q-Res+S and IR-calc are incomparable, and so are  Q-Res+S and LQU$^+$. 
 \end{prop}

For the future, the effects of adding S to more powerful proof systems 
than Q-Res remain to be investigated.

\smallskip
\textbf{Acknowledgements.}
Parts of this work were supported by the Austrian Science Fund (FWF) under grant numbers
NFN S11408-N23 (RiSE), Y464-N18, and SFB F5004.

 \bibliographystyle{plain}
 \bibliography{refs}
 
\end{document}